\documentclass[preprint]{revtex4-1}

\usepackage{natbib}
\usepackage{amsfonts}
\usepackage{latexsym,amsmath, amssymb}
\usepackage{amsthm}
\usepackage{amsthm}

\usepackage[mathscr]{eucal}
\usepackage{color}

\theoremstyle{plain}
\newtheorem{Theorem}{Theorem}

\newtheorem{Lemma}{Lemma}

\newcommand{\R}{\mathbb{R}}

\begin{document}

\title{Newtonian limit and trend to equilibrium for the\\ relativistic Fokker-Planck equation}

\author{Jos\'e Antonio Alc\'antara F\'elix}
\email{jaaf@correo.ugr.es}
\author{Simone Calogero}
\email{calogero@ugr.es}
\affiliation{Universidad de Granada\\
Departamento de Matem\'atica Aplicada \\
18071 Granada, Espa\~na}
\date{}

\begin{abstract}
The relativistic Fokker-Planck equation, in which the speed of light $c$ appears as a parameter, is considered. It is shown that in the limit $c\to\infty$ its solutions converge in $L^1$ to solutions of the non-relativistic Fokker-Planck equation, uniformly in compact intervals of time. Moreover in the case of spatially homogeneous solutions,  and provided  the temperature of the thermal bath is sufficiently small, exponential trend to equilibrium in $L^1$ is established. The dependence of the rate of convergence on the speed of light is estimated. Finally, it is proved that exponential convergence to equilibrium for all temperatures holds in a weighted  $L^2$ norm.  
\end{abstract}

\maketitle

\section{Introduction}
The Fokker-Planck equation is a widely used model to describe the dynamics of particles undergoing diffusion and friction in a surrounding fluid in thermal equilibrium~\cite{Risken}. For non-relativistic particles with mass $m>0$, and in suitable physical units, the Fokker-Planck equation is given by
\begin{equation}\label{FPC}
\partial_t f+p\cdot\nabla_xf=\Delta_pf+\frac{\theta}{m}\nabla_p\cdot(pf),\quad \theta=\frac{1}{kT}.
\end{equation}
Here $f=f(t,x,p)\geq 0$ is the one-particle distribution function in phase space; the independent variables are the time $t\geq 0$, the position $x\in \R^3$ and the momentum $p\in\R^3$ of the particles.  In the definition of the dimensional constant $\theta$, $T$ is the temperature of the thermal bath and $k$ is Boltzmann's constant.
The equilibrium state of~\eqref{FPC} is given by the Maxwellian distribution, $\mathscr{M}=\exp(-\theta |p|^2/(2m))$, up to a multiplicative constant that is fixed by the total mass of the system (which is a conserved quantity).

In this paper we consider a relativistic generalization of~\eqref{FPC}  first introduced in~\cite{DH1} by stochastic calculus methods and re-discovered later in~\cite{AC} by a different argument (see~\cite{DH2} for a review on the relativistic theory of diffusion, as well as the recent papers~\cite{haba,herr} ). In the same physical units used to write~\eqref{FPC}, the relativistic Fokker-Planck equation is given by
\begin{equation}\label{RFPc}
\partial_tf+mc\frac{p}{p^{0}}\cdot\nabla_xf=\partial_{p^i}\left[D^{ij}\partial_{p^j}f+\frac{\theta}{m} p^if\right],
\end{equation}
where $c$ denotes the speed of light, $D$ is the relativistic diffusion matrix, i.e.,
\[
D^{ij}=\frac{mc}{p^{0}}\left(\delta^{ij}+\frac{p^ip^j}{m^2c^2}\right),\quad p^{0}=\sqrt{m^2c^2+|p|^2},
\]
and where the sum of repeated indexes is, as usual, understood.
The equilibrium state of~\eqref{RFPc} is given by the J\"uttner distribution $\mathscr{J}=e^{-\theta cp^0}$, again up to a multiplicative constant.

The purpose of this paper is twofold. First we prove that~\eqref{FPC} is indeed the correct Newtonian limit of~\eqref{RFPc}; in particular we show that, as $c\to\infty$, solutions of~\eqref{RFPc} converge in $L^1$ to solutions of~\eqref{FPC}.  This provides a further justification of~\eqref{RFPc} as a meaningful relativistic generalization of~\eqref{FPC}. Our second goal is to study the trend to equilibrium for solutions of the relativistic Fokker-Planck equation. The latter problem has already been considered in~\cite{C}, where it was shown that solutions of~\eqref{RFPc} confined in a torus (i.e., $x\in\mathbb{T}^3$) converge exponentially fast in time in the $L^1$ norm to the J\"uttner equilibrium, {\it provided the temperature of the thermal bath is sufficiently small}. In this paper we study the trend to equilibrium for spatially homogeneous solutions of~\eqref{RFPc}. The assumption of spatial homogeneity allows us to derive more accurate estimates on the convergence rate. Moreover it will be shown that, at least within the class of spatially homogeneous solutions, the small temperature assumption made in~\cite{C} can be (partially) dispensed of. However in order to achieve this we have to leave the natural $L^1$ framework and prove exponential convergence in a weighted $L^2$ norm.

The Newtonian limit problem is studied in Section~\ref{newlim}; the analysis of the trend to equilibrium is carried out in Section~\ref{trend}.

\section{Newtonian limit}\label{newlim}
The main purpose of this section is to prove the following theorem.
\begin{Theorem}\label{newlimtheo}
Let $0<f,f_c\in C^2((0,\infty)\times\R^{6})$ be solutions of, respectively, eq.~\eqref{FPC} and eq.~\eqref{RFPc} with initial data $0\leq f^\mathrm{in},f_c^\mathrm{in}$. Assume that $f^\mathrm{in}_c(x,p)=0$, for $|x|>R(c)$, and $R(c)$ growing at most linearly as $c\to\infty$. Assume in addition that
\begin{equation}\label{momentder}
\Gamma_{\omega,\gamma}[f_\mathrm{in}]:=\int_{\R^{6}}\left[(1+|p|^\omega) |\nabla_xf_c^\mathrm{in}|^2+(1+|p|^\gamma) |\nabla_pf_c^\mathrm{in}|^2\right] dp\,dx<\infty,
\end{equation}
for $\gamma>7$ and $\omega>9$.
Then
$\|f_c^\mathrm{in}-f^\mathrm{in}\|_{L^1}\to 0\Rightarrow \|f_c(t)-f(t)\|_{L^1}\to 0$, as $c\to\infty$, 
uniformly on compact intervals of time.
\end{Theorem}
Throughout the paper we work with smooth solutions of~\eqref{FPC} and~\eqref{RFPc} to avoid technical difficulties.
Moreover the compact support assumption on $f^\mathrm{in}_c$ in the $x$ variable can be removed by adding suitable powers of $|x|$ inside the integral~\eqref{momentder} (which would also allow to treat the general dimension case). We prefer to sacrifice the generality of the assumptions for the benefit of a shorter and less technical proof.
 
{\bf Remark about the notation:} In the following, $A\lesssim B$ means that there exists a non-decreasing function of time (possibly a constant) $C(t)$, independent of $c>1$, such that $A\leq C(t) B$. Since we are only interested in the limiting behavior as $c\to\infty$, the assumption $c>1$ is not a restriction. 

Before proving Theorem~\ref{newlimtheo}, we show that the solution of the relativistic Fokker-Planck equation inherits the bound~\eqref{momentder} on the initial data.
\begin{Lemma}\label{moment}
If~\eqref{momentder} holds, then $\Gamma_{\omega,\gamma}[f]<\infty$, for all $\gamma,\omega\geq 0$ and $t>0$.
\end{Lemma}
\begin{proof}
1. Let $u=\nabla_x f$ and set $\beta=\theta/m$. Since each component of $u$ satisfies~\eqref{RFPc} we have
\begin{align*}
\partial_t\int_{\R^6} |p|^\omega |u|^2\,dp\,dx=&-mc\int_{\R^6}|p|^\omega\frac{p}{p^0}\cdot\nabla_x|u|^2\,dp\,dx\\
&+2\int_{\R^6}|p|^\omega\partial_{p^i}(D^{ij}\partial_{p^i}u+\beta p^iu)\cdot u\,dp\,dx.
\end{align*}
The first term in the right hand side vanishes; integrating by parts in the second term we obtain
\begin{align*}
\partial_t\int_{\R^6} |p|^\omega |u|^2\,dp\,dx=&-2\int_{\R^6} |p|^\omega D^{ij}\partial_{p^i}u\cdot\partial_{p^j}u\,dp\,dx\\
&+\beta( 3-\omega)\int_{\R^6}|p|^\omega |u|^2\,dp\,dx\\
&+\omega\int_{\R^6}\partial_{p^j}(|p|^{\omega-2}p_iD^{ij})|u|^2\,dp\,dx\\
\leq&\ \beta( 3-\omega)\int_{\R^6}|p|^\omega |u|^2\,dp\,dx\\
&+\omega\int_{\R^6}\partial_{p^j}(|p|^{\omega-2}p_iD^{ij})|u|^2\,dp\,dx
\end{align*}
where we used that $D^{ij}\partial_{p^i}u\cdot\partial_{p^j}u\leq 0$, since $D$ is positive definite. 
The result follows for $\omega=0$. Now, for the case $w\geq 2$, using that $p_iD^{ij}=(mc)^{-1}p^jp^{0}$ we have 
\[
\partial_{p^j}(|p|^{\omega-2}p_iD^{ij})=(mc)^{-1}[(\omega+1)p^{0}|p|^{\omega-2}+|p|^{\omega}(p^{0})^{-1}]\lesssim (1+|p|)|p|^{\omega-2}.
\]
Therefore
\[
\partial_t\int_{\R^6} |p|^\omega |u|^2dp\,dx\lesssim \int_{\R^6} (1+|p|^\omega)|u|^2dp\,dx
\]
and the bound on the integral of $|p|^\omega |u|^2$ follows for $\omega\geq 2$ and, by interpolation, for all $\omega\geq 0$.  

2. Now,  let $v=\nabla_p f$. Using a similar argument as above for the integral of $|p|^\gamma|v|^2$ we obtain
\begin{align*}
\partial_t\int_{\R^6} |p|^\gamma |v|^2 \,dp\,dx  \leq & \,\beta( 5-\gamma)\int_{\R^6} |p|^\gamma|v|^2 \,dp\,dx\\
&+\gamma\int_{\R^6}\partial_{p^j}(|p|^{\gamma-2}p_iD^{ij})|v|^2 \,dp\,dx-\frac{2\gamma}{mc}\int_{\R^6}\frac{|p|^{\gamma} }{p^0}|v|^2\,dp\,dx\\
&-2\int_{\R^6}|p|^{\gamma}\partial_{p^k}(D^{ij})\partial_{p^j}f\partial_{p^i}(\partial_{p^k}f)\,dp\,dx
\\
&-2mc\int_{\R^6} |p|^{\gamma}\partial_{p^k}\left(\frac{p^i}{p^0}\right)\partial_{x^i}f\partial_{p^k}f \,dp\,dx\\
=&\ I_1+I_2+I_3 +I_4+I_5 , 
\end{align*}
where we used that $\partial_{p^k}D^{ij}=(p^i\delta^{jk}+p^j\delta^{ik})/mcp^0-p^k D^{ij}/(p^0)^2$. From here,  we notice that $I_1+I_2$ can be treated as in the first part of the proof and  $I_5$ can be estimated using Cauchy's inequality
\[
I_5\leq \int |p|^{\gamma}(|v|^2+|u|^2) \,dp\,dx.
\]   
Moreover, splitting $\partial_{p^k}D^{ij}$,  integrating by parts and using that $\partial_{p^k}\left[\frac{|p|^\gamma p^k}{(p^0)^3}\right]>0$, we obtain  
\begin{align*}
I_4 =&\ \frac{1}{mc} \int_{\R^{6}}\left(2\partial_{p^k}\left[\frac{|p|^{\gamma}p^j}{p^0} \right]-\partial_{p^i}\left[\frac{|p|^\gamma p^ip^jp^k}{(p^{0})^3}\right]\right)\partial _{p^k} f  \partial _{p^j} f \,dp\,dx\\
&-mc\int_{\R^{6}}\partial_{p^k}\left[\frac{|p|^\gamma p^k}{(p^{0})^3}\right]|v|^2\,dp\,dx \\
\leq&\  \frac{1}{mc} \int_{\R^{6}}\left(2\partial_{p^k}\left[\frac{|p|^{\gamma}p^j}{p^0} \right]-\partial_{p^i}\left[\frac{|p|^\gamma p^ip^jp^k}{(p^{0})^3}\right]\right)\partial _{p^k} f  \partial _{p^j} f \,dp\,dx,
\end{align*}
Let 
\[
\Delta^{jk}=(p^0)^3 \left(2\partial_{p^k}\left[\frac{|p|^{\gamma}p^j}{p^0} \right]-\partial_{p^i}\left[\frac{|p|^\gamma p^ip^jp^k}{(p^{0})^3}\right]\right).
\]
Since $p^j\partial_{p^j} fp^k\partial_{p^k}f-|p|^{2}|v|^2\leq 0$, we have 
\begin{align*}
\Delta^{jk}\partial_{p^j}f\partial_{p^k}f-2\gamma|p|^{\gamma} (p^0)^2|v|^2&\leq 2 |p|^{\gamma-2}(p^0)^2\{|p|^2|v|^2+ \gamma[(p^k\partial_{p^k}f)^2-|p|^2 |v|^2]
\\
-(p^0)^{-2}(\gamma/2+2)|p|^2(p^k\partial_{p^k}f)^2\}&\leq 2(p^0)^2 |p|^{\gamma}|v|^2.
\end{align*}
From the last inequality we see that
\begin{align*}I_3+I_4 \leq  \frac{2}{m}\int_{\R^{6}} \frac{|p|^{\gamma}}{p^0}|v|^2\,dp\,dx,
\end{align*}
and the claim follows as in part 1.
\end{proof}

\begin{proof}[Proof of Theorem~\ref{newlimtheo}] The difference $\delta f=(f-f_c)$ is a smooth solution of
\begin{equation}\label{deltaf}
\partial_t\delta f+p\cdot\nabla_x\delta f =\frac{\theta}{m}\nabla_p\cdot(p\delta f)+\Delta_p\delta f+g_c\:,
\end{equation}
where
\[
g_c=\Delta_p f_c-\partial_{p^i}\big(D^{ij}\partial_{p^j}f_c\big)+\left[\frac{mc}{p^{0}}-1\right]p\cdot\nabla_xf_c.
\]
Let $\mathcal{F}(t,x,p,y,w)$ denote the two-point Green function of the non-relativistic Fokker-Planck equation~\eqref{FPC}. In terms of $\mathcal{F}$, the solution of~\eqref{FPC} is given by
\[
f(t,x,p)=\int_{\R^6}\mathcal{F}(t,x,p,y,w)f(0,y,w)\,dy\,dw.
\]
Since~\eqref{deltaf} reduces to~\eqref{FPC} when $g_c=0$, the  Duhamel's principle entails that the solution of~\eqref{deltaf}
can be represented as
\begin{align}
\delta f(t,x,p)=& \int_{\R^6}\mathcal{F}(t,x,p,y,w)\delta f(0,y,w)\,dw\,dy\nonumber\\
& +\int_0^t\int_{\R^6}\mathcal{F}(t-s,x,p,y,w)g_c(s,y,w)\,dw\,dy\,ds,\label{tempooo}
\end{align}
for $t\geq s$. The exact form of $\mathcal{F}$ is 
\begin{align*}
\mathcal{F}(t,x,p,y,w)=\left[\frac{\beta \exp\{\beta t\}}{4\pi\sqrt{a(2\beta,t)t-a^2(\beta,t)}}\right]^6 \exp\left\{-\frac{b(t,x,p,y,w)}{4t} \right\},
\end{align*}
where $\beta=\theta /m$,  $a(\beta,t)=\frac{\exp\{\beta t\}-1}{\beta}$ and
\begin{align*}
b(t,x,p,y,w)=&\ |\beta(x-y)+(p-w)|^2\\
&+\frac{\left|a(\beta,t)\left\{\beta(x-y)+(p-w)\right\}+t(w-p\exp \left\{\beta t\right\})\right|^2}{ta(2\beta,t)-a^2(\beta, t)},
\end{align*}
see~\cite[Eq.~(2.5)]{Car}.  Now, in the second term in the right hand side of~\eqref{tempooo} we integrate by parts {\it once} in the variable $w$ and obtain
\begin{align}
\delta f=&\int_{\R^6}\mathcal{F}(t,x,p,y,w)\delta f(0,y,w)\,dw\,dy\nonumber\\
&-\int_0^t\int_{\R^6}\nabla_w\mathcal{F}(t-s,x,p,y,w)\cdot X(f_c)(s,y,w)\,dw\,dy\,ds\nonumber\\
&+\int_0^t\int_{\R^6}\mathcal{F}(t-s,x,p,y,w)\left[\frac{mc}{w^{0}}-1\right]w\cdot\nabla_yf_c(s,y,w)\,dw\,dy\,ds,\label{repfor2}
\end{align}
where $X$ is the vector field $X^i=\partial_{w_i}-D^{ij}\partial_{w_j}$. Now we use that
\begin{align*}
\left|1-\frac{mc}{w^{0}}\right|&=\left|\frac{\sqrt{m^2c^2+|w|^2}-mc}{\sqrt{m^2c^2+|w|^2}}\right|\\
&=\frac{1}{\sqrt{m^2c^2+|w|^2}+mc}\frac{|w|^2}{\sqrt{|w|^2+m^2c^2}}\lesssim\frac{|w|^2}{c^2},
\end{align*} 
by which we also have
\begin{align*}
|X(f_c)|&\leq \sup_{i,j}|\delta^{ij}-D^{ij}||\nabla_w f_c|=\sup_{i,j}\left|\delta^{ij}\left(1-\frac{mc}{w_0}\right)-\frac{w^iw^j}{w_0mc}\right||\nabla_wf_c|\\
&\lesssim \frac{|w|^2}{c^2}|\nabla_wf_c|.
\end{align*}
Using the estimates just derived in~\eqref{repfor2} we obtain
\begin{align*}
\|\delta f_c(t)\|_{L^1}\lesssim&\int_{\R^6}\left(\int_{\R^6} \mathcal{F}(t,x,p,y,w)\,dp\,dx\right)|\delta f(0,y,w)|\,dw\,dy\nonumber\\
&+\frac{1}{c^2}\int_0^t\int_{\R^6} |w|^3|\nabla_yf_c|\left(\int_{\R^6} \mathcal{F}(t-s,x,p,y,w)\,dp\,dx\right)\,dw\,dy\,ds\nonumber\\
&+\frac{1}{c^2}\int_0^t\int_{\R^6} |w|^2|\nabla_wf_c|\left(\int_{\R^6} |\nabla_w\mathcal{F}|(t-s,x,p,y,w)\,dp\,dx\right)\,dw\,dy\,ds.
\end{align*} 
From here, we take into account that $\mathcal{F}$ has the following properties (see~\cite{Car})
\begin{align*}
\int_{\R^6} \mathcal{F}(t,x,p,y,w)\,dp\,dx&=1,\\
|\nabla_w\mathcal{F}|(t-s,x,p,y,w)&\leq \frac{C(\alpha,\beta)}{\sqrt{t-s}}\mathcal{F}(t-s,\alpha x,\alpha p,\alpha y, \alpha w)
\end{align*}
with $0<\alpha<1$ and estimate the integrals in the variables $(x,p)$ to obtain
\begin{align}\label{temporale}
\|\delta f(t)\|_{L^1}\lesssim&\|\delta f(0)\|_{L^1}+\frac{1}{c^2}\int_0^t\int_{\R^6} |w|^3|\nabla_y f_c|\,dw\,dy\,ds\nonumber\\
&+\frac{1}{c^2}\int_0^t\frac{1}{\sqrt{t-s}}\int_{\R^6} |w|^2|\nabla_w f_c|\,dw\,dy\,ds,
\end{align}
By the finite propagation speed property of the relativistic Fokker-Planck equation proved in~\cite{AC}, and the assumption that $f^\mathrm{in}_c=0$ for $|y|>R$, the solution of~\eqref{RFPc} satisfies $f_c=0$ for $|y|\geq R+ct$. Whence
\begin{align*}
\int_{\R^6}|w|^2|\nabla_w f_c|\,dw\,dy\leq& \int_{|y|\lesssim c}\int_{|w|<1}|\nabla_wf_c|\,dw\,dy\\
& +\int_{|y|\lesssim c}\int_{|w|\geq 1}|w|^2|\nabla_w f_c|\,dw\,dy\\
\lesssim&\ c^{3/2}\left(\int_{\R^{6}}|\nabla_w f_c|^2dw\,dy\right)^{1/2}\\
&+c^{3/2}\left(\int_{|w|\geq 1}|w|^{4-\gamma}dw\right)^{1/2}\left(\int_{\R^6}|w|^\gamma|\nabla_wf_c|^2dw\,dy\right)^{1/2}
\end{align*}
and so for $\gamma>7$ the integral in the left hand side is $O(c^{3/2})$. By exactly the same argument
\begin{align*}
\int_{\R^6}|w|^3|\nabla_y f_c|\,dw\,dy
\lesssim& \ c^{3/2}\left(\int_{\R^{6}}|\nabla_y f_c|^2dw\,dy\right)^{1/2}\\
&+c^{3/2}\left(\int_{|w|\geq 1}|w|^{6-\omega}dw\right)^{1/2}\left(\int_{\R^6}|w|^\omega|\nabla_yf_c|^2dw\,dy\right)^{1/2}
\end{align*}
and for $\omega>9$ the integral in the left hand side is $O(c^{3/2})$.  Using these estimates in~\eqref{temporale} we get
\[
\|\delta f(t)\|_{L^1}\lesssim\|\delta f(0)\|_{L^1}+O(1/\sqrt{c})
\]
and the theorem follows.
\end{proof}

\section{Trend to equilibrium}\label{trend}
In this section we restrict to spatially homogeneous solutions of~\eqref{RFPc}. 
Moreover for the analysis of the trend to equilibrium it is more convenient to rewrite the relativistic Fokker-Planck equation in terms of $h=f/\mathscr{J}$. We obtain
\begin{equation}
\partial_th =\partial_{p^i}\left[\frac{mc}{p^0}\left(\delta^{ij}+\frac{p^ip^j}{m^2c^2}\right)\partial_{p^j}h\right]-\frac{\theta}{m} p\cdot\nabla_ph,
\end{equation}
or equivalently,
\begin{equation}\label{FPRnew}
\partial_t h =\Delta^{(g)}_ph+Wh,
\end{equation}
where the Riemannian metric $g$ and the vector field $W$ are given by
\begin{equation}\label{metric}
g_{ij}=\frac{1}{mc}\left(p^0\delta_{ij}-\frac{p_ip_j}{p^0}\right),\quad
Wh=W^i\partial_{p^i}h,\quad W^i=-\frac{1}{m}\left(\theta+\frac{1}{2p^0c}\right)p^i
\end{equation}
and $\Delta_p^{(g)}$ denotes the Laplace-Beltrami operator of the metric $g$.
Note that $W_i=g_{ij}W^j=\partial_{p^i}\log u$, where $u$ denotes the function
\begin{equation}\label{u}
u=\frac{e^{-\theta cp^0}}{\sqrt{\det g}}=\sqrt{\frac{mc}{p^0}}e^{-\theta cp^0}.
\end{equation}
Let 
\begin{equation}\label{measure}
d\mu_\theta=Z^{-1}e^{-\theta c p^0}dp,\quad Z=\int_{\mathbb{R}^3}e^{-\theta cp^0}dp, 
\end{equation}
so that $d\mu_\theta$ is a probability measure. The reason to emphasize the dependence of the measure $\mu$ on the parameter $\theta$ will become clear soon. In the following we denote by $h$ a solution of~\eqref{FPRnew} normalized to a probability density measure:
\[
\|h\|_{L^1(d\mu_\theta)}=\int_{\mathbb{R}^3}h\,d\mu_\theta=1.
\]
This normalization can always be achieved by rescaling the solution.
The entropy functional and the entropy dissipation functional are defined by
\[
\mathfrak{D}[h]=\int_{\mathbb{R}^3}h\log h\, d\mu_\theta,\quad
\mathfrak{I}[h]=\int_{\mathbb{R}^3} g(\partial_p h,\partial_p\log h)\,d\mu_\theta,
\]
and the following entropy identity holds:
\begin{equation}\label{entropyid}
\frac{d}{dt}\mathfrak{D}[h](t)=-\mathfrak{I}[h](t).
\end{equation}
A solution of~\eqref{FPRnew} is said to converge to equilibrium in the entropic sense if $\mathfrak{D}[h]\to 0=\mathfrak{D}[1]$ as $t\to\infty$, and with exponential rate if $\mathfrak{D}[h]=O(e^{-\lambda t})$, as $t\to\infty$, for some $\lambda>0$. A sufficient condition for exponential decay of the entropy is the validity of the following logarithmic Sobolev inequality:
\begin{equation}\label{logsob}
\int_{\mathbb{R}^3}h\log h\, d\mu_\theta\leq\alpha\int_{\mathbb{R}^3} g(\partial_p h,\partial_p\log h)\,d\mu_\theta,\quad \text{for some $\alpha>0$}
\end{equation}
and for all sufficiently smooth probability densities measure $h$ (not necessarily solutions of~\eqref{FPRnew}).
In fact using~\eqref{logsob} in~\eqref{entropyid} we obtain
\[
\frac{d}{dt}\mathfrak{D}[h]\leq-\frac{1}{\alpha}\mathfrak{D}[h]\Rightarrow\mathfrak{D}[h]\lesssim\exp(-t/\alpha).
\]
The Cisz\'ar-Kullback inequality, $\|h-1\|_{L^1(d\mu_\theta)}\leq\sqrt{2\mathfrak{D}}$, see~\cite{CS},  implies  that $h$ converges to equilibrium in $L^1(d\mu_\theta)$ with exponential rate $(2\alpha)^{-1}$, or equivalently, the solution of~\eqref{FPRnew} satisfies
\begin{equation}\label{expconv}
\|f(t)-\mathscr{J}_M\|_{L^1(dp)}\lesssim e^{-t/(2\alpha)},
\end{equation} 
where $\mathscr{J}_M$ denotes the J\"uttner equilibrium with mass $M=\|f\|_{L^1(dx)}$. Clearly,~\eqref{expconv} provides the most natural notion of convergence to equilibrium for solutions to the relativistic Fokker-Planck equation.
 
Thus the question of exponential trend to equilibrium in $L^1$ has been reduced to prove that~\eqref{logsob} holds.
\begin{Theorem}\label{expdecayen}
The logarithmic Sobolev inequality~\eqref{logsob} holds for $\theta>\theta_0=\frac{7}{2mc^2}$, for a constant $\alpha$ given by
\[
\frac{1}{2\alpha}=\left\{
\begin{array}{ll}
\mathscr{P}(mc)=\frac{2\theta mc^2-7}{2mc^2},& \text{if} \ \theta_0<\theta\leq\frac{4}{mc^2},\\
\\
\mathscr{P}\left(\frac{2}{13}\theta mc^2+\frac{mc}{13}\sqrt{4\theta^2m^2c^4-39}\right),&\text{if}\ \theta>\frac{4}{mc^2},
\end{array}\right.
\]
where $\mathscr{P}(x)$ is the rational function
\[
\mathscr{P}(x)=\frac{2\theta c x^3-13x^2+2\theta m^2c^3x-m^2c^2}{4mcx^3}.
\]
\end{Theorem}
\begin{proof} The proof is carried out by using the Bakry-Emery curvature bound condition~\cite{BE,B} which states that~\eqref{logsob} holds provided the tensor $\widetilde{\mathrm{Ric}}=\mathrm{Ric}-\nabla^2_p\log u$ ---
called the Bakry-Emery-Ricci tensor --- satisfies 
$\widetilde{\mathrm{Ric}}\geq\frac{1}{2\alpha}g$. In the definition of $\widetilde{\mathrm{Ric}}$, $\mathrm{Ric}$ is the Ricci tensor of $g$, while $u$ is the function~\eqref{u}.
In our case  the Bakry-Emery-Ricci tensor reads
\[
\widetilde{\mathrm{Ric}}_{ij}=-\frac{1}{4(p^0)^2}(1+4c\theta (p^0))\delta_{ij}+\frac{6\theta c (p^0)^3-12 (p^0)^2+2\theta m^2c^3 p^0-m^2c^2}{4mc(p^0)^3}g_{ij}.
\]
Now we use
\[
g(X,X)= \frac{p^0}{mc}\left(|X|^2-\frac{(p\cdot X)^2}{(p^0)^2}\right)\geq\frac{mc}{p^0}|X|^2,\quad\text{for all $X\in\R^3$}
\]
and so
\[
\widetilde{\mathrm{Ric}}(X,X)\geq \left[\frac{1}{4mc(p^0)^3}\left(2\theta c(p^0)^3-13(p^0)^2+2\theta m^2c^3p^0-m^2c^2\right)\right]g(X,X).
\]
The function on square brackets is $\mathscr{P}(p_0)$. It is easy to show that
$\min \{\mathscr{P}(p^0), p^0\geq mc\}$ is strictly positive if and only if $\theta>\theta_0$. The value of $(2\alpha)^{-1}$ is obtained by looking for the minimum of $\mathscr{P}$ on $[mc,\infty)$.  
\end{proof}
The condition $\theta>\theta_0$ means that the previous result holds only for small temperatures of the thermal bath, since $\theta\sim T^{-1}$. To prove exponential decay of the entropy for all temperatures one needs to find a substitute for the Bakry-Emery curvature bound condition used in the proof of Theorem~\ref{expdecayen}. Although there are several criteria in the literature for the validity of logarithmic Sobolev inequalities,  we were unable to find one that applies in our situation. Thus we proceed by a different approach. Since the following argument is independent of the dimension, we consider~\eqref{FPRnew} with $p\in\R^N$.  Let us consider, instead of the entropy $\mathfrak{D}[h]$,  the new functional $\mathfrak{L}[h]=\|h\|_{L^2(d\mu_\theta)}^{2}$. Computing the time derivative of $\mathfrak{L}[h-1]$ we obtain
\[
\frac{d}{dt}\mathfrak{L}[h-1](t)=-2\int_{\R^N}g(\partial_ph,\partial_ph)\,d\mu_\theta.
\]
Thus $\mathfrak{L}[h-1]$ decays exponentially, i.e., $h\to 1$ in $L^2(d\mu_\theta)$ exponentially fast, if we show that the following Poincar\'e inequality
\begin{equation}\label{poincareineq}
\int_{\R^N}(h-1)^2d\mu_\theta\leq \lambda\int_{\R^N}g(\partial_ph,\partial_ph)\,d\mu_\theta,\quad\text{for some $\lambda>0$},
\end{equation}
holds for all sufficiently smooth probability densities measure $h$.
The validity of the Poincar\'e inequality~\eqref{poincareineq} is equivalent to the existence of a spectral gap for the operator in the right hand side of~\eqref{FPRnew}, which will now be established by applying a criterion due to Wang, see~\cite{wang}.   
To adhere with the notation in~\cite{wang}, let us rewrite~\eqref{FPRnew} in the form
\begin{equation}\label{FPRnew2}
\partial_th=a^{ij}\partial_{p^i}\partial_{p^j}h+b^j\partial_{p^j}h,\quad t>0,\quad p\in\R^N,
\end{equation}
where
\[
a^{ij}=\frac{mc}{\sqrt{m^2c^2+|p|^2}}\left(\delta^{ij}+\frac{p^ip^j}{m^2c^2}\right),\quad b^j=\left(\frac{Np^j}{mc\sqrt{m^2c^2+|p|^2}}-\frac{\theta}{m}p^j\right).
\]
For $r>0$ define 
\[
\gamma(r)=\sup_{|p|=r}\frac{r[\mathrm{Tr}(a(p))+p\cdot b(p)]}{a^{ij}p_ip_j}-\frac{1}{r},\quad C(r)=\int_1^r\gamma(s)ds,\quad \alpha(r)=\inf_{|p|=r}\frac{a^{ij}p_ ip_ j}{r^2}.
\]
Then by~\cite[Th.3.1]{wang}, the spectral gap for the operator in the right hand side of~\eqref{FPRnew2} is strictly positive provided there exists a function $y\in C([1,\infty))$ such that $\sup_{t\geq 1}G_y(t)<\infty$, where
\[
G_y(t)=\frac{1}{y(t)}\int_1^te^{-C(r)}\int_r^\infty e^{C(s)}\frac{y(s)}{\alpha(s)}\,ds\,dr.
\]
\begin{Theorem}\label{spectralgap}
The Poincar\'e inequality~\eqref{poincareineq} holds for all $\theta>0$.
\end{Theorem}
\begin{proof}
For eq.~\eqref{FPRnew2} the function $G(t)$ is given by
\[
G_y(t)=\frac{mc}{y(t)}\int_1^t\frac{e^{\theta c\sqrt{m^2c^2+r^2}}}{r^{N-1}\sqrt{m^2c^2+r^2}}\int_r^\infty e^{-\theta c\sqrt{m^2c^2+s^2}}s^{N-1}y(s)\,ds\,dr.
\]
Let $\beta<\theta c$ and pick $y(t)=\frac{e^{\beta t} }{t^{N-1}}$. After straightforward estimates we obtain
\[
G_y(t)\leq\frac{mc}{\theta c-\beta}e^{\theta c(\sqrt{m^2c^2+1}-1)}\underbrace{\frac{t^{N-1}}{e^{\beta t}}\int_1^t\frac{e^{\beta r}}{r^N}dr}_{F(t)}.
\] 
Since $\lim_{t\to\infty} F(t)=0$, the result by Wang applies and the theorem is proved. 
\end{proof}

{\bf Note:} While this paper was being written, we have been informed by J.~Angst that he was also able to prove the Poincar\'e inequality~\eqref{poincareineq} and therefore the exponential convergence to equilibrium in $L^2(d\mu_\theta)$ for solutions of~\eqref{FPRnew}. The proof by Angst~\cite{angst} employs a criterion for the existence of a spectral gap to elliptic operators established in~\cite{bakry}. 

{\bf Acknowledgments:} The first author is sponsored by the Mexican National Council for Science and Technology (CONACYT) with scholarship number 214152. The second author worked on this paper during a long term visit to the Center of Mathematics for the Applications (CMA) in Oslo.

\end{document}